\documentclass[10pt,conference]{IEEEtran}
\usepackage{newtxtext}
\usepackage{cite}
\usepackage{amsmath,amssymb,amsthm,mathrsfs}
\usepackage{graphicx}
\usepackage{booktabs}
\usepackage{tikz}
\usepackage{balance}
\usepackage{geometry}
\usetikzlibrary{arrows.meta,positioning}
\makeatletter
\renewcommand\small{\@setfontsize\small{10}{11.5}}
\renewcommand\footnotesize{\@setfontsize\footnotesize{10}{11.5}}
\g@addto@macro\normalsize{%
  \setlength\abovedisplayskip{2.5pt}%
  \setlength\belowdisplayskip{2.5pt}% 
  \setlength\abovedisplayshortskip{1.5pt}%
  \setlength\belowdisplayshortskip{1.5pt}}
\makeatother

\makeatletter
\renewcommand\section{\@startsection{section}{1}{\z@}%
  {1.0ex plus 1.0ex minus 0.4ex}%
  {0.5ex plus 0.8ex minus 0ex}%
  {\normalfont\normalsize\centering\scshape}}
\renewcommand\subsection{\@startsection{subsection}{2}{\z@}%
  {1.0ex plus 1.0ex minus 0.4ex}%
  {0.5ex plus 0.5ex minus 0ex}%
  {\normalfont\normalsize\itshape}}
\makeatother

\makeatletter
\newtheoremstyle{ieeetight}
  {2pt plus 1pt minus 1pt}
  {2pt plus 1pt minus 1pt}
  {\itshape}
  {}
  {\normalfont\bfseries}
  {.}
  {0.5em}
  {}
\theoremstyle{ieeetight}
\newtheorem{theorem}{Theorem}
\newtheorem{proposition}{Proposition}
\newtheorem{corollary}{Corollary}
\newtheorem{lemma}{Lemma}
\newtheorem{definition}{Definition}
\makeatother

\newcommand{\E}{\mathbb E}
\newcommand{\R}{\mathbb R}
\newcommand{\T}{\mathsf T}
\newcommand{\Var}{\operatorname{Var}}

\newcommand{\tr}{\operatorname{tr}}
\newcommand{\rank}[1]{\operatorname{rank}\!\left( #1 \right)}
\newcommand{\spn}{\operatorname{span}}
\newcommand{\defeq}{\triangleq}
\newcommand{\cA}{\mathcal A}
\newcommand{\cX}{\mathcal X}
\newcommand{\cS}{\mathcal S}
\newcommand{\cY}{\mathcal Y}
\newcommand{\cF}{\mathcal F}

\newcommand{\cZ}{\mathcal Z}
\newcommand{\cV}{\mathcal V}
\newcommand{\cW}{\mathcal W}
\newcommand{\lpmi}{\mathfrak i}
\newcommand{\pim}{\pi_{\mathsf M}}
\newcommand{\pic}{\pi_{\mathsf C}}
\newcommand{\pii}{\pi_{\mathsf I}}
\newcommand{\Hs}{\mathscr H}
\newcommand{\MSE}{\mathsf{MSE}}
\newcommand{\md}[1]{\zeta_{\leq #1}}

\newcommand{\spc}[1]{\cF_{#1}}
\newcommand{\spcn}[2]{\cF_{#1}^{#2}}
\newcommand{\spcs}[2]{\cF_{#1|#2}}
\newcommand{\spcsn}[3]{\cF_{#1|#2}^{#3}}
\newcommand{\proj}[2]{\Pi\!\left(#1;\, #2\right)}
\newcommand{\norm}[1]{\left\| #1 \right\|}
\newcommand{\orthp}{\mathbin{\boxplus}} %
\newcommand{\orthm}{\mathbin{\boxminus}} %
\newcommand{\fb}{\bar{f}}
\newcommand{\gb}{\bar{g}}

\newcommand{\La}{\Lambda}
\newcommand{\bbnorm}[1]{\left\| #1 \right\|}
\renewcommand{\mid}{\mathrel{\mkern-2mu\vert\mkern-2mu}}

\begin{document}
% Activate the IEEEtran.bst control entry (suppresses URL/DOI in references);
% must appear before the first \cite so the settings apply to all entries.
\bstctlcite{IEEEtran:BSTcontrol}
\title{
Universal Distributed Feature Extraction Complementary to Side Information via H-Score
}

\author{
\IEEEauthorblockN{Yuan Xinjie}
\IEEEauthorblockA{\textit{Tsinghua Shenzhen International} \\
\textit{Graduate School, Tsinghua University}\\
Shenzhen, China \\
yuanxj23@mails.tsinghua.edu.cn}
\and
\IEEEauthorblockN{Dier Tang}
\IEEEauthorblockA{\textit{Department of Mathematics} \\
\textit{The University of Hong Kong}\\
Hong Kong, China \\
tangde\_math@connect.hku.hk}
\and
\IEEEauthorblockN{Shao-Lun Huang}
\IEEEauthorblockA{\textit{Tsinghua Shenzhen International} \\
\textit{Graduate School, Tsinghua University}\\
Shenzhen, China \\
twn2gold@gmail.com}
}

\maketitle

\begin{abstract}
This paper considers a distributed estimation system where edge nodes transmit averaged features to a center node for parameter estimation. A key challenge arises when the center node possesses auxiliary side information: how to design local extractors that are universally optimal for any downstream task, while avoiding redundancy with the side information. To tackle this problem, we orthogonally decompose the joint function space and demonstrate in the weak-dependence regime that the modal decomposition of the canonical dependence kernel's conditional-additive component yields optimal, task-independent features. We construct a distributed nested H-score to learn this decomposition and show that, under a limited feature budget, the mean-squared error (MSE) gap to the best task-specific features is determined solely by the task's alignment with the discarded, least informative modes, and vanishes to leading order for every square-integrable task once the budget is sufficient.
\end{abstract}

\section{Introduction}\label{sec:intro}

Distributed systems like integrated sensing and communication (ISAC)~\cite{liu_isac_2022,ahmadipour_information_2024,xinjie_unbounded_2025}, IoT, and sensor networks typically comprise edge nodes and a center node. Due to rate and energy constraints, edge nodes transmit block averages of low-dimensional features rather than raw observations. The center node then estimates latent parameters, often aided by side information unavailable at the edge.

For such systems, a primary question is how to design local feature extraction mappings and central estimators. This has long been studied in multi-terminal inference \cite{tenney_detection_1981,ahlswede_hypothesis_1986,han_statistical_1998}, where local messages are optimized for a fixed task. However, existing designs face two limitations. First, task-oriented designs~\cite{shao_learning_2022,gunduz_beyond_2023}, e.g., supervised learning with a mean-squared-error (MSE) loss, are optimal only for the training task and must be redesigned whenever the task changes. Second, applying universal features \cite{xu_information_2021b,tang_sharp_2026} directly is sub-optimal with side information $S$: local features capturing information already in $S$ are redundant, wasting communication resources. A universal feature-extraction scheme is thus desired that is complementary to $S$ and independent of the specific task.

Two further ingredients inform our design. On the physical layer, computation over multiple-access channels \cite{nazer_computation_2007} and over-the-air aggregation \cite{gastpar_uncoded_2008,zhu_broadband_2020} directly deliver sums of node-side preprocessings---an additive structure we formalize at the level of statistical dependence; on the learning side, the H-score objective underlying our approach originates from the centralized Soft-HGR framework \cite{wang_efficient_2019}. None of the above lines, however, meets both requirements at once.

In this paper, we design such universal features from an information-geometric perspective, building upon the feature geometry framework\cite{xu_neural_2024b}. Operating in the weak-dependence regime, we first orthogonally decompose the joint function space into a Markov subspace, a conditional subspace, and an interaction subspace. By projecting the canonical dependence kernel (CDK) onto them, we successfully isolate the dependence that is neither captured by the side information nor invisible to distributed implementations. We then show that a distributed nested H-score simultaneously realizes this subspace decomposition and the modal decomposition of the conditional-additive component, whose left modes provide the desired task-independent local features. Finally, we derive the center-node estimator based on these features and characterize its exact leading-order MSE gap to the best task-specific features: the gap is determined by the task's alignment with the discarded modes and vanishes for every square-integrable task once all modes are retained.

\section{System Model and Preliminaries}\label{sec:model}
\subsection{System Model}\label{subsec:system}
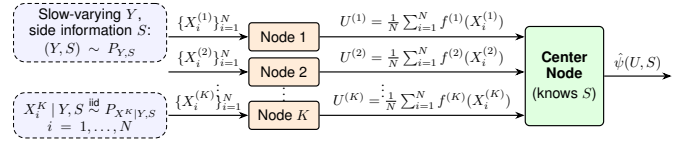
\begin{figure}[!t]
\centering
\resizebox{0.99\columnwidth}{!}{%
\begin{tikzpicture}[
        >=Stealth,
        font=\sffamily\small,
        subnode/.style={
            rectangle, draw, thick,
            minimum width=1.7cm, minimum height=0.68cm,
            text centered, rounded corners=2pt,
            fill=orange!15
        },
        centernode/.style={
            rectangle, draw, thick,
            minimum width=1.9cm, minimum height=2.4cm,
            text centered, rounded corners=3pt,
            align=center,
            fill=green!12
        },
        arrow/.style={-Stealth, thick},
    ]

    % ===== Left: Parameters =====
    \node[draw, densely dashed, rounded corners=8pt, fill=blue!5,
            text width=3.5cm, align=center, inner sep=3pt] (params) at (-2.5, 4.5)
            {Slow-varying $Y$,\\ side information $S$:\\[0.1em] $(Y,S) \sim P_{Y,S}$};

    \node[draw, densely dashed, rounded corners=8pt, fill=blue!5,
            text width=3.5cm, align=center, inner sep=3pt] (datagen) at (-2.5, 2.5)
            {$X_i^K \mid Y,S \overset{\text{iid}}{\sim} P_{X^K|Y,S}$\\$i=1,\ldots,N$};

    % ===== Nodes =====
    \node[subnode] (N1) at (2.2, 4.4) {Node 1};
    \node[subnode] (N2) at (2.2, 3.55) {Node 2};
    \node[font=\footnotesize] at (2.2, 3.12) {$\vdots$};
    \node[font=\footnotesize] at (0.6, 3.2) {$\vdots$};
    \node[font=\footnotesize] at (4.6, 3.2) {$\vdots$};
    \node[subnode] (NK) at (2.2, 2.5) {Node $K$};

    % ===== Center node =====
    \node[centernode] (CN) at (9.0, 3.45) {\textbf{Center}\\\textbf{Node}\\[0.2em] {\footnotesize\textcolor{black}{(knows $S$)}}};

    % ===== Data arrows =====
    \draw[arrow] (-0.6, 4.4) -- (N1.west)
        node[above, midway, font=\footnotesize] {$\{X_i^{(1)}\}_{i=1}^N$};
    \draw[arrow] (-0.6, 3.55) -- (N2.west)
        node[above, midway, font=\footnotesize] {$\{X_i^{(2)}\}_{i=1}^N$};
    \draw[arrow] (-0.6, 2.5) -- (NK.west)
        node[above, midway, font=\footnotesize] {$\{X_i^{(K)}\}_{i=1}^N$};

    % ===== Feature arrows =====
    \draw[arrow] (N1.east) -- (CN.west |- N1.east)
        node[above, midway, font=\footnotesize] {$U^{(1)}=\frac{1}{N}\sum_{i=1}^{N}f^{(1)}(X_{i}^{(1)})$};
    \draw[arrow] (N2.east) -- (CN.west |- N2.east)
        node[above, midway, font=\footnotesize] {$U^{(2)}=\frac{1}{N}\sum_{i=1}^{N}f^{(2)}(X_{i}^{(2)})$};
    \draw[arrow] (NK.east) -- (CN.west |- NK.east)
        node[above, midway, font=\footnotesize] {$U^{(K)}=\frac{1}{N}\sum_{i=1}^{N}f^{(K)}(X_{i}^{(K)})$};

    % ===== Output arrow =====
    \draw[arrow] (CN.east) -- (11.6, 3.45)
        node[above, midway, font=\footnotesize] {$\hat{\psi}(U, S)$};

\end{tikzpicture}%
}
\caption{System model: node $k$ transmits the block average $U^{(k)}$ of its local features, and the center node, which alone knows $S$, forms the estimate $\hat{\psi}(U,S)$.}
\label{fig:system-model}
\end{figure}

We consider a distributed estimation system comprising $K$ nodes and a single center node, as illustrated in Fig.~\ref{fig:system-model}, where the objective is to estimate a target function $\psi(Y)$ of a latent parameter $Y$ at the center node, based on locally extracted features and a global side information variable $S$ available exclusively at the center node. Let $Y \in \cY$ and $S \in \cS$ denote the latent parameter and the side information. We adopt a block-constant model wherein $(Y, S)$ remains invariant over blocks of $N$ time slots, with $N$ large enough, and is treated as fixed in the analysis below. Let $X^K = (X^{(1)}, \dots, X^{(K)})^{\T} \in \cX^{K} \defeq \cX^{(1)} \times \cdots \times \cX^{(K)}$ denote the joint observation, governed by a joint distribution $P_{X^{K}, S, Y}$. Conditioned on $(Y,S) = (y,s)$, the nodes observe $N$ independent and identically distributed (i.i.d.) samples $\{X_i^K\}_{i=1}^{N}$ from $P_{X^K|Y=y,S=s}$, and each node $k \in \{1, \dots, K\}$ observes only its local component sequence $\{X_i^{(k)}\}_{i=1}^{N}$.  

Subject to communication constraints, the nodes cannot transmit raw observations. Instead, each node $k$ applies a local feature mapping $f^{(k)} : \cX^{(k)} \to \R^{d}$, where $d$ denotes the per-node feature dimension budget, and transmits the block average
\begin{equation}
U^{(k)} = \frac{1}{N}\sum_{i=1}^{N} f^{(k)}(X_i^{(k)}) \in \R^{d}.
\label{eq:block-average}
\end{equation}
The center node concatenates the transmissions into the stacked vector $U = (U^{(1)\T}, \dots, U^{(K)\T})^{\T} \in \R^{Kd}$ and forms an estimate $\hat{\psi}(U, S)$ of $\psi(Y)$ as an arbitrary measurable function of $(U,S)$, measured by the MSE $\E[(\hat{\psi}(U, S) - \psi(Y))^2]$. Without loss of generality, all local feature mappings $\{f^{(k)}\}$ and the target function $\psi$ are assumed zero-mean, since constant offsets can be absorbed into the center estimator. This paper addresses the joint design problem
\begin{equation}
\min_{\substack{f^{(k)} \in \spcn{\cX^{(k)}}{d},\, k \in [K]\\ \hat\psi:\, \R^{Kd}\times\cS \to \R}} \E\bigl[(\hat{\psi}(U, S) - \psi(Y))^2\bigr],
\label{eq:design-problem}
\end{equation}
where $[K] \defeq \{1,\dots,K\}$ and $\spcn{\cZ}{d}$ denotes the set of $d$-dimensional square-integrable vector functions on $\cZ$ (formalized below). Our goal is feature designs that are \emph{universal}, i.e., independent of the task $\psi$, yet simultaneously optimal for all tasks, and \emph{complementary to the side information}, i.e., spending no feature dimensions on dependence already explained by $S$.

\textit{Remark:} This model captures applications such as ISAC, where edge nodes (e.g., dual-functional terminals) observe signals $X_i^{(k)}$ from the base station, and the center node combines side information $S$ (e.g., ambient temperature) with $U$ to estimate a target function $\psi(Y)$ (e.g., target velocity) rather than the high-dimensional $Y$ itself. Transmitting block averages also meets the low-latency requirements of ISAC, which preclude buffering long sequences, and aligns with gradient aggregation in distributed learning.

\subsection{Feature Geometry and Canonical Dependence Kernels}\label{subsec:geometry}
We work in the feature geometry on $\spc{\cX^K \times \cS \times \cY}$ with the metric distribution $R_{X^K, S, Y} \defeq P_{X^K} \, P_{S, Y}$. Here $\spc{\cZ} \defeq L^2(R_Z)$ denotes the space of square-integrable functions on a measurable space $\cZ$ under a probability measure $R_Z$, with inner product $\langle h_1, h_2 \rangle \defeq \E_{Z \sim R_Z}[h_1(Z) h_2(Z)]$ and norm $\norm{h} = \sqrt{\langle h, h\rangle}$; $\spcn{\cZ}{d} \defeq (\spc{\cZ})^d$. For $f \in \spc{\cX}$ and $g \in \spc{\cY}$ the tensor product is $(f \otimes g)(x, y) \defeq f(x) g(y)$; for vector-valued $f = (f_1,\dots,f_d)^{\T}$, $g = (g_1,\dots,g_d)^{\T}$, write $f \otimes g \defeq \sum_{i=1}^{d} f_i \otimes g_i$; for subspaces $\cV$, $\cW$, $\cV \otimes \cW \defeq \spn\{f \otimes g : f \in \cV, g \in \cW\}$. $\proj{h}{\cV}$ denotes the orthogonal projection of $h$ onto a closed subspace $\cV$ (elementwise for vector-valued $h$), $\cV \orthp \cW$ the orthogonal direct sum, and $\cV \orthm \cW$ the orthogonal complement of $\cW$ in $\cV$. A local function $f^{(k)} \in \spc{\cX^{(k)}}$ is identified with its embedding in $\spc{\cX^K}$, $x^K \mapsto f^{(k)}(x^{(k)})$.

Under the metric distribution, the induced metric on $\spc{\cX^K}$ is $P_{X^K}$ and that on $\spc{\cS \times \cY}$ is $P_{S, Y}$. Since conditional expectation is an orthogonal projection, the side space admits the decomposition $\spc{\cS \times \cY} = \spc{\cS} \orthp \spcs{\cY}{\cS}$, where
\begin{equation}
\spcs{\cY}{\cS} \defeq \{h \in \spc{\cS \times \cY} : \E[h(S, Y) \mid S] = 0\},
\label{eq:conditional-side-space}
\end{equation}
$\spc{\cS}$ is identified with its embedding $\{h(s,y) \equiv h(s)\}$, and $\spcsn{\cY}{\cS}{d} \defeq (\spcs{\cY}{\cS})^d$.

The \emph{joint CDK} of $(X^K; S, Y)$ is
\begin{equation}
\lpmi_{X^K; S, Y}(x^K, s, y) \defeq \frac{P_{X^K, S, Y}(x^K, s, y)}{P_{X^K}(x^K) P_{S, Y}(s, y)} - 1.
\label{eq:joint-cdk}
\end{equation}
To isolate the effect of the side information, we invoke the orthogonal decomposition of function spaces \cite[Prop.~22]{xu_neural_2024b}, which splits the joint CDK orthogonally into the \emph{Markov component} $\lpmi_{X^K; S} \in \spc{\cX^K} \otimes \spc{\cS}$ and the \emph{conditional dependence component} $\lpmi_{X^K; Y|S} \in \spc{\cX^K} \otimes \spcs{\cY}{\cS}$:
\begin{align}
\lpmi_{X^K; S, Y} &= \lpmi_{X^K; S} + \lpmi_{X^K; Y|S}, \label{eq:cdk-split}\\
\norm{\lpmi_{X^K; S, Y}}^2 &= \norm{\lpmi_{X^K; S}}^2 + \norm{\lpmi_{X^K; Y|S}}^2, \notag
\end{align}
where
\begin{align*}
\lpmi_{X^K; S}(x^K, s) &= \frac{P_{X^K|S=s}(x^K)}{P_{X^K}(x^K)} - 1,\\
\lpmi_{X^K; Y|S}(x^K, s, y) &= \frac{P_{X^K|S=s,Y=y}(x^K) - P_{X^K|S=s}(x^K)}{P_{X^K}(x^K)}.
\end{align*}

Throughout, we analyze the weak-dependence regime
\[
\sup_{x^K, s, y} \bigl| \lpmi_{X^K; S, Y}(x^K, s, y) \bigr| = O(\epsilon),
\]
which both components in \eqref{eq:cdk-split} inherit, as an average and a difference of the joint CDK.

Writing $\E[f\mid s] \defeq \E[f(X^K) \mid S=s]$ and $\E[f\mid s,y] \defeq \E[f(X^K)\mid S=s, Y=y]$, the following pairing identities, which connect the CDKs with conditional expectations, hold for any $f \in \spc{\cX^K}$ (and hence for every embedded local function $f^{(k)}$):
\begin{align}
\bigl\langle \lpmi_{X^K; S}(\cdot, s),\, f \bigr\rangle_{P_{X^K}} &= \E[f\mid s] - \E[f],
\label{eq:markov-repr}\\
\bigl\langle \lpmi_{X^K; Y|S}(\cdot, s, y),\, f \bigr\rangle_{P_{X^K}} &= \E[f\mid s,y] - \E[f\mid s].
\label{eq:conditional-repr}
\end{align}

\subsection{Modal Decomposition and H-Score}\label{subsec:modal}
Suppose $\gamma \in \spc{\cX \times \cY}$ admits a finite modal decomposition
\begin{equation}
\gamma = \sum_{i=1}^{\rank{\gamma}} \zeta_i(\gamma) = \sum_{i=1}^{\rank{\gamma}} \sigma_i \, (f_i^* \otimes g_i^*),
\label{eq:modal-decomposition}
\end{equation}
with singular values $\sigma_1 \geq \sigma_2 \geq \cdots > 0$, orthonormal modes $\langle f_i^*, f_j^* \rangle = \langle g_i^*, g_j^* \rangle = \delta_{ij}$, and $\rank{\gamma}$ the number of nonzero singular values. The truncation $\zeta_{\leq i}(\gamma) \defeq \sum_{j \leq i} \zeta_j(\gamma)$, with $\zeta_j(\gamma) \defeq 0$ for $j > \rank{\gamma}$, is the optimal rank-$i$ approximation, unique if $\sigma_i > \sigma_{i+1}$ (Eckart--Young--Mirsky)~\cite{xu_neural_2024b}. When $\gamma$ is the CDK of a pair of variables, the modes are their Hirschfeld--Gebelein--R\'enyi maximal correlation functions~\cite{xu_maximal_2020}.

The H-score for features $f \in \spcn{\cX}{d}$, $g \in \spcn{\cY}{d}$ is
\begin{align}
\Hs(f, g) &= \tfrac12\bigl(\norm{\lpmi_{X;Y}}^2 - \norm{\lpmi_{X;Y} - f \otimes g}^2\bigr) \label{eq:hscore}\\
&= \E_{P_{X,Y}}[f^{\T}(X)g(Y)] - \tfrac12 \tr(\La_f \La_g), \notag
\end{align}
with $\La_f = \E_{P_X}[f f^{\T}]$, $\La_g = \E_{P_Y}[g g^{\T}]$, where $\lpmi_{X;Y}$ denotes the CDK of a generic pair $(X;Y)$; its maximizers satisfy $f \otimes g = \zeta_{\leq d}(\lpmi_{X;Y})$, and its nested version sums the H-scores of the coordinate prefixes to yield the ordered modes with automatic orthogonality \cite{xu_neural_2024b}.

\section{Additive Geometry and Triple Orthogonal Decomposition}\label{sec:geometry}
In this section, we first characterize how the center node estimates with given features, which motivates the additive space, and then establish the orthogonal decomposition of the joint function space based on it.

\subsection{Estimation with Given Feature Extractors}\label{subsec:fixed}
Let $f \defeq (f^{(1)\T}, \dots, f^{(K)\T})^{\T}: \cX^K \to \R^{Kd}$ denote the stacked feature map, so that $U = \frac{1}{N}\sum_{i=1}^{N} f(X_i^K)$. Define the feature covariance matrix $\La_f \in \R^{Kd \times Kd}$, the conditional feature variation $\beta(y,s) \in \R^{Kd}$, and the task--feature correlation $c_{\psi}(s) \in \R^{Kd}$ as
\begin{align}
\La_f &\defeq \E_{P_{X^K}}[f f^{\T}], \quad
\beta(y,s) \defeq \E[f\mid s,y] - \E[f\mid s], \notag\\
c_{\psi}(s) &\defeq \E_{Y|S=s}[\psi(Y) \beta(Y, s)],
\label{eq:fixed-moments}
\end{align}
where $\La_f$ is assumed nonsingular on the retained feature span, which can always be arranged by removing redundant coordinates. Since $N$ is large, we model $U \mid (Y,S)=(y,s)$ as Gaussian $\mathcal{N}\bigl(\E[f\mid s,y], \tfrac1N \La_f\bigr)$, neglecting the $O(\epsilon)$ dependence of the covariance on $(s,y)$, whose effect on the posterior is a factor $\sqrt{N}$ smaller than that of the mean. We then have the following characterization; the proof is deferred to Appendix~\ref{app:proof-fixed}.

\begin{lemma}\label{lemma:optimal-estimator-fixed}
For any given local features $f$, the posterior distribution of $Y$ at the center node is approximated by
\begin{equation}
p_{Y|U,S}(y\mid u,s) = p_{Y|S}(y\mid s)\bigl[1 + N\,u^{\T}\La_f^{-1}\beta(y,s)\bigr] + O(\epsilon^2).
\label{eq:posterior-leading}
\end{equation}
Consequently, for any $\psi \in \spc{\cY}$, the conditional expectation of $\psi(Y)$ given $U$ and $S$ is approximated by
\begin{equation}
\E[\psi(Y) \mid U, S] = \E[\psi(Y) \mid S] + N \cdot c_{\psi}(S)^{\T} \La_f^{-1} \, U + O(\epsilon^2).
\label{eq:cond-expectation-leading}
\end{equation}
\end{lemma}

Thus, under weak dependence, both the posterior of $Y$ and the conditional expectation of $\psi(Y)$ are linear in $U$. Hence, for given extractors $f$, the minimum mean-squared error (MMSE) estimator of $\psi(Y)$ at the center node is the leading-order conditional expectation 
\begin{equation}
      \begin{aligned}
            \hat{\psi}^*(U,S;f) = \E[\psi(Y) \mid S] + N \cdot c_{\psi}(S)^{\T} \La_f^{-1} U,
      \end{aligned}
      \label{eq:task-specific-estimator}
\end{equation}
and the corresponding achievable MSE is
\begin{align}
      % \small
      & \MSE_\psi(f) \defeq  \min_{\hat{\psi}}\E\bigl[(\hat\psi(U,S) - \psi(Y))^2\bigr]   \label{eq:mse-fixed-features}\\
      &\quad  =   \E[\Var(\psi(Y) \mid S)] - N \cdot \E\bigl[c_{\psi}(S)^{\T} \La_f^{-1} c_{\psi}(S)\bigr] + o(\epsilon^2).\notag
\end{align}
% \begin{equation}
%       \small
%       \begin{aligned}
%             \MSE_\psi(f) \defeq& \min_{\hat{\psi}}\E\bigl[(\hat\psi(U,S) - \psi(Y))^2\bigr] \\
%             =& \E[\Var(\psi(Y) \mid S)] - N \cdot \E_{S}\bigl[c_{\psi}(S)^{\T} \La_f^{-1} c_{\psi}(S)\bigr] + o(\epsilon^2),
%       \end{aligned}
%       \label{eq:mse-fixed-features}
% \end{equation}
Alternatively, by defining the \emph{task image}
\begin{equation}
    \begin{aligned}
          \bar{\psi}(x^K, s) \defeq \E_{Y|S=s}\bigl[\psi(Y)\, \lpmi_{X^K;Y|S}(x^K, s, Y)\bigr] \in \spc{\cX^K \times \cS},
          \label{eq:bar-psi-def}
    \end{aligned}
\end{equation}
the MSE admits the projection form
\begin{equation}
      \begin{aligned}
      &\MSE_\psi(f) = \E[\Var(\psi(Y) \mid S)] \\
      & \qquad - N \cdot \bigl\|\proj{\bar{\psi}}{\spn\{f\} \otimes \spc{\cS}}\bigr\|_{P_{X^K} P_S}^2 + o(\epsilon^2),
      \end{aligned}
      \label{eq:mse-projection}
\end{equation}
where $\spn\{f\} = \spn\{f_i : i \in [Kd]\} \subseteq \spc{\cX^K}$ denotes the span of all stacked components; see Appendix~\ref{app:proof-projection}.

Since the center node estimates by linearly combining the received features, the performance of a feature design depends only on the subspace spanned by the local features. This motivates the \emph{additive space}
\begin{equation}
\cA \defeq \spc{\cX^{(1)}} + \cdots + \spc{\cX^{(K)}} = \overline{\textstyle\bigcup_{f} \spn\{f\}} \;\subseteq\; \spc{\cX^K},
\label{eq:additive-space}
\end{equation}
where the local spaces are embedded in $\spc{\cX^K}$ as above, and the closure is taken over all finite-dimensional feasible subspaces. Moreover, for any side space $\cV \subseteq \spc{\cS \times \cY}$, $\sum_{k=1}^{K} \bigl(\spc{\cX^{(k)}} \otimes \cV\bigr) = \cA \otimes \cV \subseteq \spc{\cX^K \times \cS \times \cY}$.

\subsection{Triple Orthogonal Decomposition}\label{subsec:triple}
Combining the orthogonal decomposition of the $X^K$-side, $\spc{\cX^K} = \cA \orthp \bigl(\spc{\cX^K} \orthm \cA\bigr)$, with the side-space decomposition \eqref{eq:conditional-side-space}, the full space splits into three mutually orthogonal subspaces:
\begin{equation}
\begin{split}
\spc{\cX^K \times \cS \times \cY} \;=\; & \underbrace{\bigl(\cA \otimes \spc{\cS}\bigr)}_{\text{Markov subspace}} \;\orthp\; \underbrace{\bigl(\cA \otimes \spcs{\cY}{\cS}\bigr)}_{\text{conditional subspace}} \\
& \orthp\; \underbrace{\bigl((\spc{\cX^K} \orthm \cA) \otimes \spc{\cS \times \cY}\bigr)}_{\text{interaction subspace}}.
\end{split}
\label{eq:three-subspaces}
\end{equation}
The first two pieces form the \emph{accessible subspace} $\cA \otimes \spc{\cS \times \cY}$ realizable under the communication constraint; the third piece is invisible to any distributed implementation. Projecting the joint CDK onto these three subspaces gives the following decomposition.

\begin{definition}\label{def:components}
The \emph{Markov-additive}, \emph{conditional-additive}, and \emph{interaction components} of the joint CDK are its projections onto the three subspaces in \eqref{eq:three-subspaces}:
\begin{equation*}
\begin{aligned}
\pim &\defeq \proj{\lpmi_{X^K; S, Y}}{\cA \otimes \spc{\cS}} = \proj{\lpmi_{X^K; S}}{\cA \otimes \spc{\cS}}, \\
\pic &\defeq \proj{\lpmi_{X^K; S, Y}}{\cA \otimes \spcs{\cY}{\cS}} = \proj{\lpmi_{X^K; Y|S}}{\cA \otimes \spcs{\cY}{\cS}}, \\
\pii &\defeq \lpmi_{X^K; S, Y} - \pim - \pic.
\end{aligned}
% \label{eq:def-components}
\end{equation*}
\end{definition}

The alternative expressions of $\pim$ and $\pic$ are due to the split \eqref{eq:cdk-split}: $\lpmi_{X^K; Y|S} \perp \cA \otimes \spc{\cS}$ and $\lpmi_{X^K; S} \perp \cA \otimes \spcs{\cY}{\cS}$. Moreover, $\pii = \proj{\lpmi_{X^K; S, Y}}{\bigl(\spc{\cX^K} \orthm \cA\bigr) \otimes \spc{\cS \times \cY}}$. By the orthogonality of \eqref{eq:three-subspaces}, the three components are mutually orthogonal, i.e.,
\begin{equation}
      \begin{aligned}
            \lpmi_{X^K;S,Y} &= \pim + \pic + \pii, \\
            \bbnorm{\lpmi_{X^K; S, Y}}^2 &= \bbnorm{\pim}^2 + \bbnorm{\pic}^2 + \bbnorm{\pii}^2. 
            \label{eq:pythagorean}
      \end{aligned}
\end{equation}

Intuitively, $\pim$ is dependence that the center node already knows through $S$, hence useless for estimating $\psi(Y)$; $\pic$ is the dependence on $Y$ beyond $S$ that additive features can extract, hence crucial for estimation; and $\pii$ is invisible to any distributed implementation.

\section{Modal Decomposition with Nested H-Score}\label{sec:main}
This section first shows how the modal decomposition of the two accessible components is obtained by a distributed nested H-score, and then develops the center-node estimator and its MSE gap.

\subsection{Deriving the Modal Decomposition via Nested H-Scores}\label{subsec:nested}
Since $\pii$ is invisible to any distributed implementation, we decompose only $\pim$ and $\pic$, with feature functions of dimensions $\bar{d}$ and $\tilde{d}$:
\begin{equation}
\fb \in \cA^{\bar{d}}, \quad \gb \in \spcn{\cS}{\bar{d}}, \qquad
\tilde{f} \in \cA^{\tilde{d}},\quad \tilde{g} \in \spcn{\cS \times \cY}{\tilde{d}},
\label{eq:def-fb-tilde}
\end{equation}
where $\cA^{d}$ denotes the set of $d$-dimensional vector functions whose coordinates lie in $\cA$, i.e., $\fb = \sum_{k=1}^{K} \fb^{(k)}$ with $\fb^{(k)} \in \spcn{\cX^{(k)}}{\bar{d}}$, and similarly for $\tilde{f}$. The left branches are subject to the distributed (additive) constraint, whereas the right branches, realized at the center node during offline training, are unconstrained. Following the nesting technique of \cite{xu_neural_2024b}, we define $\Hs_{\mathrm{MC}}$ as the sum of two H-scores:
\begin{equation}
\Hs_{\mathrm{MC}}\biggl( \begin{bmatrix}\fb\\ \tilde{f}\end{bmatrix},  \begin{bmatrix}\gb\\ \tilde{g}\end{bmatrix}  \biggr)  \;\defeq\; \Hs(\fb,\gb) + \Hs\biggl( \begin{bmatrix}\fb\\ \tilde{f}\end{bmatrix},  \begin{bmatrix}\gb\\ \tilde{g}\end{bmatrix}  \biggr),
\label{eq:hscore-mc}
\end{equation}
where the first H-score is between $(X^K;S)$ and the second between $(X^K; S, Y)$, both under the metric distribution $R_{X^K,S,Y}$. The corresponding nested H-score network structure is shown in Fig.~\ref{fig:nested-hscore-network}.

\begin{figure}[!t]
\centering
\includegraphics[width=0.95\columnwidth]{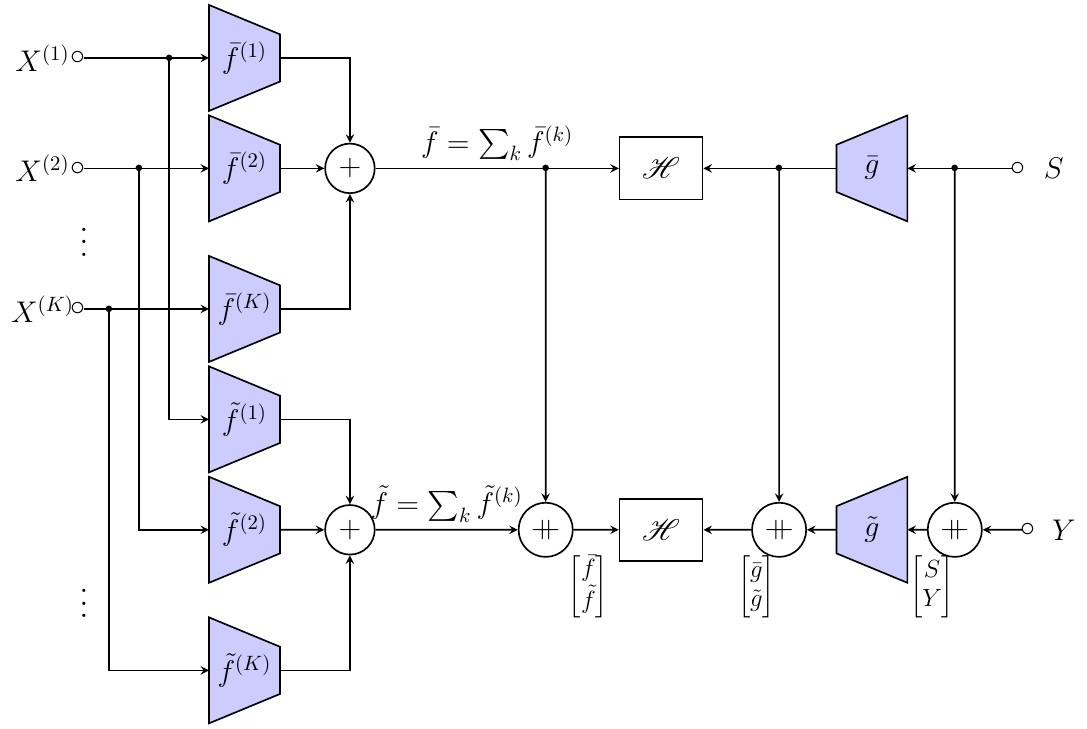}
\caption{The nested H-score network for the modal decomposition of the Markov-additive component $\pim$ and the conditional-additive component $\pic$.}
\label{fig:nested-hscore-network}
\vspace{-1ex}
\end{figure}

Maximizing $\Hs_{\mathrm{MC}}$ over $\fb, \gb, \tilde{f}, \tilde{g}$ yields the modal decomposition of both components, as formalized below (proof in Appendix~\ref{app:proof-nested}).

\begin{proposition}\label{prop:modal-decomposition}
For $\fb,\gb, \tilde{f}, \tilde{g}$ defined in \eqref{eq:def-fb-tilde} and $\bar{d} \geq \rank{\pim}$, the score \eqref{eq:hscore-mc} is maximized if and only if
\begin{equation}
\fb \otimes \gb = \pim, \quad
% \tilde{f} \otimes \tilde{g} = \md{\tilde{d}} (\pic),
 \tilde{f} \otimes \tilde{g} =  \begin{cases}
   \md{\tilde{d}} (\pic), & \tilde{d}<\rank{\pic},\\    \pic, & \tilde{d} \geq \rank{\pic}.
   \end{cases}
\label{eq:modal-decomposition-condition}
\end{equation}
\end{proposition}

\subsection{Estimation with the Modal Decomposition}\label{subsec:estimation}
Let $\pic = \sum_{i=1}^{r} \sigma_i\, (f_i^* \otimes g_i^*)$, $r \defeq \rank{\pic}$, be the modal decomposition \eqref{eq:modal-decomposition} of $\pic$. For a given budget $\tilde d$, let $\tilde{f}^* = \sum_{k=1}^{K} \tilde{f}^{(k)*} \in \cA^{\tilde{d}}$ and $\tilde{g}^* \in \spcsn{\cY}{\cS}{\tilde{d}}$ be maximizers of $\Hs_{\mathrm{MC}}$, so that, by Proposition~\ref{prop:modal-decomposition},
\begin{equation}
\tilde{f}^* \otimes \tilde{g}^* = \begin{cases}
\md{\tilde{d}}(\pic), & \tilde{d} < r,\\
\pic, & \tilde{d} \geq r.
\end{cases}
\label{eq:modal-decomposition-pic}
\end{equation}
Each node $k$ uses $f^{(k)} = \tilde{f}^{(k)*}$ (so $d = \tilde d$), and the center node uses the aggregate
\[
V \defeq \sum_{k=1}^{K} U^{(k)} = \frac{1}{N}\sum_{i=1}^{N} \tilde f^*(X_i^K) \in \R^{\tilde d},
\]
which over-the-air aggregation~\cite{gastpar_uncoded_2008,zhu_broadband_2020} delivers directly. Lemma~\ref{lemma:optimal-estimator-fixed} applies verbatim to the feature map $\tilde f^*$, with $\tilde\beta(y,s) \defeq \E[\tilde f^*\mid s,y] - \E[\tilde f^*\mid s] = \langle \tilde f^*, \lpmi_{X^K;Y|S}(\cdot,s,y)\rangle_{P_{X^K}}$ by \eqref{eq:conditional-repr}. Since the coordinates of $\tilde f^*$ lie in $\cA$, only $\pic$ contributes to this inner product; moreover, $\pic - \tilde f^*\otimes\tilde g^* = \sum_{i>\tilde d}\sigma_i (f_i^*\otimes g_i^*)$ is orthogonal to the coordinates of $\tilde f^*$, whose span is $\spn\{f_1^*,\dots,f_{\tilde d}^*\}$ when $\tilde d < r$. Hence
\begin{equation}
\tilde\beta(y,s) = \La_{\tilde f^*}\,\tilde g^*(y,s), \qquad \La_{\tilde f^*} \defeq \E\bigl[\tilde f^* \tilde f^{*\T}\bigr],
\label{eq:lambda-beta}
\end{equation}
which gives the following characterization of the distributed estimator.

\begin{proposition}\label{prop:distributed-posterior}
With the local feature extractors $\tilde f^{(k)*}$, the posterior distribution and the conditional expectation at the center node are
\begin{align}
p_{Y|V,S}(y\mid v,s) &= p_{Y|S}(y\mid s)\bigl[ 1 + N v^{\T} \tilde{g}^{*}(y,s) \bigr] + O(\epsilon^2),
\label{eq:distributed-posterior}\\
\E[\psi(Y) \mid V, S] &= \E[\psi(Y) \mid S] + N \cdot V^{\T} \tilde{c}_{\psi}(S) + O(\epsilon^2),
\label{eq:distributed-expectation}
\end{align}
where $\tilde{c}_{\psi}(S) \defeq \E\bigl[\tilde{g}^{*}(Y,S)\,\psi(Y) \,\big|\, S\bigr]$. Hence the MMSE estimator of $\psi$ in the leading order is
\begin{equation}
\hat{\psi}^{*}(V,S) = \E[\psi(Y) \mid S] + N \cdot V^{\T} \tilde{c}_{\psi}(S),
\label{eq:optimal-estimator-distributed}
\end{equation}
and its MSE for any task $\psi \in \spc{\cY}$ is
\begin{equation}
\begin{split}
\MSE_{\psi}^{\mathrm{univ}}
=\;& \E\bigl[\Var(\psi(Y)\mid S) \bigr] \\
&- N \,\E\bigl[\tilde{c}_{\psi}(S)^{\T} \La_{\tilde{f}^{*}}  \tilde{c}_{\psi}(S) \bigr] + o(\epsilon^2).
\end{split}
\label{eq:mse-univ-prop4}
\end{equation}
\end{proposition}
\begin{proof}
Substitute $\La_{\tilde f^*}^{-1}\tilde\beta = \tilde g^*$ from \eqref{eq:lambda-beta} into \eqref{eq:posterior-leading} and \eqref{eq:cond-expectation-leading}, and substitute $c_{\psi} = \La_{\tilde f^*}\tilde{c}_{\psi}$, which follows from \eqref{eq:lambda-beta}, into \eqref{eq:mse-fixed-features}.
\end{proof}

Once the modes are learned offline, e.g., by maximizing the empirical nested H-score, the node-side maps $\tilde{f}^{(k)*}$ are fixed and task-independent: when the task $\psi$ changes, only the coefficients $\tilde{c}_{\psi}$ and the baseline $\E[\psi(Y)\mid S]$ are recomputed at the center node, without retraining the nodes.

\subsection{Optimality and Universality of the Feature Extractors}\label{subsec:universal}
Let $\MSE_{\psi}^{*} \defeq \inf_{f} \MSE_\psi(f)$ be the best MSE achievable by task-specific extractors of arbitrary dimension. Since $\spn\{f\}$ ranges over all finite-dimensional subspaces of $\cA$, \eqref{eq:mse-projection} and the closure characterization \eqref{eq:additive-space} give
\begin{equation}
\begin{split}
\MSE_{\psi}^{*} =\;& \E[\Var(\psi(Y)\mid S)] \\
&- N \cdot \bigl\|\proj{\bar{\psi}}{\cA \otimes \spc{\cS}}\bigr\|_{P_{X^K} P_S}^2 + o(\epsilon^2),
\end{split}
\label{eq:single-task-optimal-mse-proj}
\end{equation}
attained when $\spn\{f\}$ contains the left singular space of $\proj{\bar{\psi}}{\cA \otimes \spc{\cS}}$. Comparing \eqref{eq:single-task-optimal-mse-proj} with \eqref{eq:mse-univ-prop4} gives the gap.

\begin{theorem}\label{thm:universality-gap}
For any budget $\tilde d$ and any task $\psi \in \spc{\cY}$, the universal features in \eqref{eq:modal-decomposition-pic} satisfy
\begin{equation}
\MSE_{\psi}^{\mathrm{univ}} - \MSE_{\psi}^{*} = N \sum_{i=\tilde{d}+1}^{r} \sigma_i^2\, \E\bigl[\rho_{\psi,i}(S)^2\bigr] + o(\epsilon^2),
\label{eq:mse-gap}
\end{equation}
where $\rho_{\psi,i}(s) \defeq \E[g_i^*(Y,s)\psi(Y)\mid S=s]$ is the conditional correlation of the task with the $i$-th right mode.
\end{theorem}
\begin{proof}
Write $\lpmi_{X^K;Y|S} = \pic + \pi_{\mathsf{IC}}$ with $\pi_{\mathsf{IC}} \defeq \proj{\lpmi_{X^K;Y|S}}{(\spc{\cX^{K}} \orthm \cA) \otimes \spcs{\cY}{\cS}}$. In \eqref{eq:bar-psi-def}, the contribution of $\pi_{\mathsf{IC}}$ lies in $(\spc{\cX^{K}} \orthm \cA) \otimes \spc{\cS}$ and is annihilated by the projection onto $\cA \otimes \spc{\cS}$, whereas that of $\pic$ already lies in $\cA \otimes \spc{\cS}$. Hence
\begin{align*}
\proj{\bar{\psi}}{\cA \otimes \spc{\cS}}(x^K,s) &= \bigl\langle \pic(x^K, s, \cdot), \psi \bigr\rangle_{P_{Y|S=s}}\\
&= \textstyle\sum_{i=1}^{r} \sigma_i\, \rho_{\psi,i}(s)\, f_i^*(x^K),
\end{align*}
whose squared norm is $\sum_{i=1}^{r} \sigma_i^2\, \E[\rho_{\psi,i}(S)^2]$ by the orthonormality of $\{f_i^*\}$. Similarly, the definition of $\tilde c_\psi$ and \eqref{eq:modal-decomposition-pic} give
\begin{align*}
\tilde f^{*\T}(x^K)\,\tilde c_\psi(s) &= \E\bigl[(\tilde f^*\otimes\tilde g^*)(x^K,s,Y)\,\psi(Y) \bigm| S=s\bigr]\\
&= \textstyle\sum_{i=1}^{\min\{\tilde d, r\}} \sigma_i\, \rho_{\psi,i}(s)\, f_i^*(x^K),
\end{align*}
so that
\[
\E\bigl[\tilde{c}_{\psi}(S)^{\T} \La_{\tilde{f}^{*}} \tilde{c}_{\psi}(S)\bigr] = \sum_{i=1}^{\min\{\tilde d, r\}} \sigma_i^2\, \E[\rho_{\psi,i}(S)^2].
\]
Substituting into \eqref{eq:single-task-optimal-mse-proj} and \eqref{eq:mse-univ-prop4} and subtracting yields \eqref{eq:mse-gap}.
\end{proof}

Thus, the penalty of a limited budget is the task's $\sigma_i^2$-weighted energy on the discarded right modes. Since the nested H-score discards the modes with the smallest singular values, the loss is small unless the task concentrates on these least informative modes; in particular, any task with $\rho_{\psi,i} \equiv 0$ for all $i > \tilde d$ is estimated optimally even when $\tilde d < r$. When all modes are retained, the sum in \eqref{eq:mse-gap} is empty:

\begin{corollary}\label{cor:universality}
If $\tilde{d} \geq r$, then for every task $\psi \in \spc{\cY}$,
\[
\MSE_{\psi}^{\mathrm{univ}} = \MSE_\psi^{*} + o(\epsilon^2),
\]
i.e., the features are simultaneously optimal for all tasks.
\end{corollary}

Since $\pim$ is absorbed by $(\fb,\gb)$, which are used only in training and never transmitted, no transmitted dimension is spent on dependence already explained by $S$. Without side information, i.e., for constant $S$, we have $\pim = 0$ and $\pic = \proj{\lpmi_{X^K;Y}}{\cA \otimes \spc{\cY}}$, so the scheme reduces to universal distributed features without side information.

\section{Conclusion}\label{sec:conclusion}
We studied universal local feature design for distributed estimation with side information at the center node. Under weak dependence, a triple orthogonal decomposition shows that only the conditional-additive component of the CDK is useful. Its modal decomposition, learned by a distributed nested H-score without knowledge of the task, yields features whose MSE gap to the best task-specific features is determined by the task's alignment with the discarded modes and vanishes when all modes are retained. We believe this offers a principled basis for reusable, side-information-aware feature design in distributed sensing.

\appendices

\section{Proof of Lemma~\ref{lemma:optimal-estimator-fixed}}\label{app:proof-fixed}
Under the Gaussian model of Section~\ref{subsec:fixed}, with $\mu_f(y,s) \defeq \E[f(X^K)\mid s,y]$,
\begin{align*}
& p_{Y|U,S}(y\mid u,s) \propto \\ & p_{Y|S}(y\mid s) \cdot \exp\Bigl( - \tfrac{N}{2}\bigl(u-\mu_f(y,s)\bigr)^{\T} \La_{f}^{-1} \bigl(u - \mu_f(y,s)\bigr)\Bigr).
\end{align*}
Expanding the quadratic form and using $\mu_f(y,s) = \E[f\mid s] + \beta(y,s)$ with $\beta(y,s) = O(\epsilon)$ and $\E[f\mid s] = O(\epsilon)$ (by \eqref{eq:markov-repr} and the weak-dependence bound), the $y$-dependent part of the exponent is $N\,u^{\T}\La_f^{-1}\beta(y,s) + O(\epsilon^2)$, since $\mu_f^{\T}\La_f^{-1}\mu_f = O(\epsilon^2)$. Exponentiating to first order gives \eqref{eq:posterior-leading}, and integrating $\psi(y)$ against it under $P_{Y|S=s}$ yields \eqref{eq:cond-expectation-leading}. For the MSE, the orthogonality principle conditioned on $S$ gives
\[
\MSE_\psi(f) = \E[\Var(\psi(Y)\mid S)] - \E\bigl[\Var(\E[\psi(Y)\mid U,S]\mid S)\bigr],
\]
where, by \eqref{eq:cond-expectation-leading},
\begin{align*}
&\Var(\E[\psi(Y)\mid U,S]\mid S)\\
&\quad= N^2 c_{\psi}(S)^{\T}\La_f^{-1}\,\E[UU^{\T}\mid S]\,\La_f^{-1}c_{\psi}(S) + o(\epsilon^2),
\end{align*}
with $\E[UU^{\T}\mid S] = \tfrac1N\La_f + O(\epsilon^2)$ (using $\E[U \mid S] = O(\epsilon)$ and the law of total variance). Substitution gives \eqref{eq:mse-fixed-features}.

\section{Proof of the Projection Form \eqref{eq:mse-projection}}\label{app:proof-projection}
Define $\eta(x^K, s) \defeq f(x^K)^{\T} \La_f^{-1} c_{\psi}(s) \in \spn\{f\} \otimes \spc{\cS}$, whose squared norm under $P_{X^K}P_S$ is $\E_{S}\bigl[ c_{\psi}(S)^{\T} \La_f^{-1} c_{\psi}(S) \bigr]$, precisely the correction term in \eqref{eq:mse-fixed-features}. It remains to show $\eta = \proj{\bar{\psi}}{\spn\{f\}\otimes \spc{\cS}}$; since the tensors $f_i \otimes c$ ($i \in [Kd]$, $c \in \spc{\cS}$) span this space, it suffices to verify $\langle \bar{\psi} - \eta, f_i \otimes c \rangle = 0$. By \eqref{eq:bar-psi-def}, Fubini, and \eqref{eq:conditional-repr},
\[
\E_{P_{X^K}}[\bar{\psi}(\cdot,s)\, f_i] = \E_{Y|S=s}\bigl[\psi(Y)\, \beta_i(Y,s)\bigr] = c_{\psi,i}(s),
\]
and $\E_{P_{X^K}}[\eta(\cdot,s) f_i] = (\La_f \La_f^{-1} c_{\psi}(s))_i = c_{\psi,i}(s)$. Hence $\eta$ satisfies the defining orthogonality of the projection, and substituting its norm into \eqref{eq:mse-fixed-features} proves \eqref{eq:mse-projection}.

\section{Proof of Proposition~\ref{prop:modal-decomposition}}\label{app:proof-nested}
Let $\bar\gamma \defeq \fb \otimes \gb$ and $\tilde\gamma \defeq \tilde{f} \otimes \tilde{g}$, and write $F \defeq (\fb^{\T}, \tilde f^{\T})^{\T}$, $G \defeq (\gb^{\T}, \tilde g^{\T})^{\T}$. Then $\bar\gamma \in \cA \otimes \spc{\cS}$, and the projections of $\tilde\gamma$ onto the subspaces in \eqref{eq:three-subspaces} are $\tilde\gamma_{\mathsf I} = 0$, $\tilde\gamma_{\mathsf M} = \sum_{j} \tilde{f}_j \otimes \proj{\tilde{g}_j}{\spc{\cS}}$, and $\tilde\gamma_{\mathsf C} = \sum_{j} \tilde{f}_j \otimes \proj{\tilde{g}_j}{\spcs{\cY}{\cS}}$, the latter two of rank at most $\tilde{d}$. Using the geometric form \eqref{eq:hscore} of the H-score and the orthogonality of \eqref{eq:cdk-split} and \eqref{eq:three-subspaces},
\begin{align}
\Hs(\fb,\gb) &= \tfrac12\norm{\pim}^2 - \tfrac12\norm{\pim - \bar\gamma}^2,
\label{eq:proof-nested-h1}\\
\Hs(F, G) &= \tfrac12\bigl(\norm{\pim}^2 + \norm{\pic}^2\bigr) \notag\\
&\quad - \tfrac12\Bigl[\norm{\pim - \bar\gamma - \tilde\gamma_{\mathsf M}}^2 + \norm{\pic - \tilde\gamma_{\mathsf C}}^2\Bigr],
\label{eq:proof-nested-h2}
\end{align}
where \eqref{eq:proof-nested-h1} uses $\lpmi_{X^K;S} - \pim \perp \pim - \bar\gamma$, and $\pii$ cancels in \eqref{eq:proof-nested-h2} since $\bar\gamma + \tilde\gamma$ is orthogonal to the interaction subspace. Hence maximizing $\Hs_{\mathrm{MC}}$ is equivalent to minimizing
\[
\norm{\pim - \bar\gamma}^2 + \norm{\pim - \bar\gamma - \tilde\gamma_{\mathsf M}}^2 + \norm{\pic - \tilde\gamma_{\mathsf C}}^2,
\]
whose three nonnegative terms attain their minima simultaneously iff $\bar\gamma = \pim$ (feasible since $\bar d \geq \rank{\pim}$), $\tilde\gamma_{\mathsf M} = 0$, and $\tilde\gamma_{\mathsf C} = \md{\tilde d}(\pic)$ (by Eckart--Young--Mirsky), i.e., iff \eqref{eq:modal-decomposition-condition} holds.

% ============================================================
% Per the ITW submission policy, the AI-use declaration and the
% references are the only content allowed on page 6 (the main
% text including appendices is limited to 5 pages).
% ============================================================

\section*{Declaration on the Use of Artificial Intelligence}
In this work, the authors used DeepSeek-V4-Pro to assist with the proofs, and used GLM-5.3 and Kimi-K3 to assist with polishing each section of the paper. All AI-assisted content was carefully reviewed and edited by authors, who take full responsibility for the integrity and accuracy of this publication.

\balance

\bibliographystyle{IEEEtran}
\bibliography{ref}

@IEEEtranBSTCTL{IEEEtran:BSTcontrol,
  CTLuse_url = "no",
  CTLuse_doi = "no"
}

@inproceedings{xinjie_unbounded_2025,
  title = {Unbounded {{Multi-Letter Marginal Gain}} for {{Bistatic Integrated Sensing}} and {{Communication}}},
  booktitle = {2025 {{Asia Pacific Workshop}} on {{Data Science}} and {{Information Theory}} ({{APWDSIT}})},
  author = {Yuan, Xinjie and Peng, Tianren and Fu, Haohao and Huang, Shao-Lun},
  year = 2025,
  month = oct,
  pages = {1--6},
  doi = {10.1109/APWDSIT64675.2025.11478556},
  url = {https://ieeexplore.ieee.org/abstract/document/11478556}
}

@article{xu_maximal_2020,
  title = {Maximal {{Correlation Regression}}},
  author = {Xu, Xiangxiang and Huang, Shao-Lun},
  year = 2020,
  journal = {IEEE Access},
  volume = {8},
  pages = {26591--26601},
  issn = {2169-3536},
  doi = {10.1109/ACCESS.2020.2971386},
  url = {https://ieeexplore.ieee.org/document/8979352/}
}

@article{xu_neural_2024b,
  title = {Neural {{Feature Learning}} in {{Function Space}}},
  author = {Xu, Xiangxiang and Zheng, Lizhong},
  year = 2024,
  journal = {Journal of Machine Learning Research},
  volume = {25},
  number = {142},
  pages = {1--76},
  issn = {1533-7928},
  url = {http://jmlr.org/papers/v25/23-1202.html}
}

@inproceedings{xu_information_2021b,
  title = {An {{Information Theoretic Framework}} for {{Distributed Learning Algorithms}}},
  booktitle = {2021 {{IEEE International Symposium}} on {{Information Theory}} ({{ISIT}})},
  author = {Xu, Xiangxiang and Huang, Shao-Lun},
  year = 2021,
  month = jul,
  pages = {314--319},
  publisher = {IEEE},
  address = {Melbourne, Australia},
  doi = {10.1109/ISIT45174.2021.9517827},
  url = {https://ieeexplore.ieee.org/document/9517827/},
  copyright = {https://doi.org/10.15223/policy-029},
  isbn = {978-1-5386-8209-8},
  langid = {english}
}

@article{tenney_detection_1981,
  title = {Detection with distributed sensors},
  author = {Tenney, Robert R. and Sandell, Jr., Nils R.},
  journal = {IEEE Transactions on Aerospace and Electronic Systems},
  volume = {AES-17},
  number = {4},
  pages = {501--510},
  year = {1981}
}

@article{ahlswede_hypothesis_1986,
  title = {Hypothesis testing with communication constraints},
  author = {Ahlswede, Rudolf and Csisz{\'a}r, Imre},
  journal = {IEEE Transactions on Information Theory},
  volume = {IT-32},
  number = {4},
  pages = {533--542},
  year = {1986}
}

@article{han_statistical_1998,
  title = {Statistical inference under multiterminal data compression},
  author = {Han, Te Sun and Amari, S.-I.},
  journal = {IEEE Transactions on Information Theory},
  volume = {44},
  number = {6},
  pages = {2300--2324},
  year = {1998},
  doi = {10.1109/18.720540}
}

@article{nazer_computation_2007,
  title = {Computation over multiple-access channels},
  author = {Nazer, Bobak and Gastpar, Michael},
  journal = {IEEE Transactions on Information Theory},
  volume = {53},
  number = {10},
  pages = {3498--3516},
  year = {2007}
}

@article{gastpar_uncoded_2008,
  title = {Uncoded transmission is exactly optimal for a simple {{Gaussian}} ``sensor'' network},
  author = {Gastpar, Michael},
  journal = {IEEE Transactions on Information Theory},
  volume = {54},
  number = {11},
  pages = {5247--5251},
  year = {2008}
}

@article{zhu_broadband_2020,
  title = {Broadband analog aggregation for low-latency federated edge learning},
  author = {Zhu, Guangxu and Wang, Yong and Huang, Kaibin},
  journal = {IEEE Transactions on Wireless Communications},
  volume = {19},
  number = {1},
  pages = {491--506},
  year = {2020}
}

@article{shao_learning_2022,
  title = {Learning task-oriented communication for edge inference: An information bottleneck approach},
  author = {Shao, Jiawei and Mao, Yuyi and Zhang, Jun},
  journal = {IEEE Journal on Selected Areas in Communications},
  volume = {40},
  number = {1},
  pages = {197--211},
  year = {2022}
}

@article{gunduz_beyond_2023,
  title = {Beyond transmitting bits: {{Context}}, semantics, and task-oriented communications},
  author = {G{\"u}nd{\"u}z, Deniz and others},
  journal = {IEEE Journal on Selected Areas in Communications},
  volume = {41},
  number = {1},
  pages = {5--41},
  year = {2023}
}

@inproceedings{wang_efficient_2019,
  title = {An efficient approach to informative feature extraction from multimodal data},
  author = {Wang, Lichen and Wu, Jiaxiang and Huang, Shao-Lun and Zheng, Lizhong and Xu, Xiangxiang and Zhang, Lin and Huang, Junzhou},
  booktitle = {Proceedings of the AAAI Conference on Artificial Intelligence},
  volume = {33},
  number = {1},
  pages = {5281--5288},
  year = {2019},
  doi = {10.1609/aaai.v33i01.33015281}
}

@article{liu_isac_2022,
  author  = {Liu, Fan and Cui, Yuanhao and Masouros, Christos and Xu, Jie and Han, Tony Xiao and Eldar, Yonina C. and Buzzi, Stefano},
  title   = {Integrated Sensing and Communications: Toward Dual-Functional Wireless Networks for {6G} and Beyond},
  journal = {IEEE Journal on Selected Areas in Communications},
  volume  = {40},
  number  = {6},
  pages   = {1728--1767},
  year    = {2022}
}

@article{ahmadipour_information_2024,
  author  = {Ahmadipour, Mehrasa and Kobayashi, Mari and Wigger, Mich{\`e}le and Caire, Giuseppe},
  title   = {An Information-Theoretic Approach to Joint Sensing and Communication},
  journal = {IEEE Transactions on Information Theory},
  volume  = {70},
  number  = {2},
  pages   = {1124--1146},
  year    = {2024}
}

@misc{tang_sharp_2026,
  title = {Sharp {{Regret Bounds}} and a {{Task-Covariance Correction}} for {{Spectral Representation Learning}}},
  author = {Tang, D. and Tan, J. Y. and Han, G.},
  howpublished = {arXiv:2609.15825},
  year = 2026,
  month = sep
}
\end{document}